\documentclass[envcountsame]{llncs}
\pdfoutput=1

\usepackage{llncsdoc}
\usepackage[utf8]{inputenc}
\usepackage{amsmath}
\usepackage{amsfonts}
\usepackage{amssymb}
\usepackage{graphicx}
\usepackage[ruled, vlined]{algorithm2e}
\usepackage{multicol}
\usepackage{todonotes}
\usepackage{multirow}
\usepackage{subfigure}
\usepackage{numprint}
\usepackage{wrapfig}
\usepackage{color}

\newcommand{\C}{\ensuremath{\zeta}}

\newcommand{\X}{\ensuremath{\mathcal{X}}}
\newcommand{\Set}{\ensuremath{\mathcal{M}}}
\newcommand{\R}{\ensuremath{\mathcal{R}}}

\newcommand{\Xer}[2]{\ensuremath{\mathcal{X}_{#1 #2}}}

\newcommand{\mipShort}{minIP}

\title{Graph Clustering with Surprise: Complexity and Exact Solutions\thanks{This work was partially supported by the DFG under grant WA 654/19-1}}
\author{Tobias Fleck, Andrea Kappes \and Dorothea Wagner}
\institute{Institute of Theoretical Informatics, Karlsruhe Institute of Technology, Germany} 

\begin{document}
\maketitle

\begin{abstract}
Clustering graphs based on a comparison of the number of links within clusters and the expected value of this quantity in a random graph has gained a lot of attention and popularity in the last decade. 
Recently, Aldecoa and Mar{\'\i}n proposed a related, but slightly different approach leading to the quality measure \emph{surprise}, and reported good behavior in the context of synthetic and real world benchmarks.
We show that the problem of finding a clustering with optimum surprise is $\mathcal{NP}$-hard.
Moreover, a bicriterial view on the problem permits to compute optimum solutions for small instances by solving a small number of integer linear programs, and leads to a polynomial time algorithm on trees.
\end{abstract}
\section{Introduction}
\label{sec:introduction}
\emph{Graph clustering}, i.e., the partitioning of the entities of a network into densely connected groups, has received growing attention in the literature of the last decade, with applications ranging from the analysis of social networks to recommendation systems and bioinformatics~\cite{f-c-09}.
Mathematical formulations thereof abound; for an extensive overview on different approaches see for example the reviews of Fortunato~\cite{f-c-09} and Schaeffer~\cite{s-gc-07}.

One line of research that recently gained a lot of popularity is based on \emph{null models}, the most prominent objective function in this context being the \emph{modularity} of a clustering~\cite{ng-fecsn-04}.
Roughly speaking, the idea behind this approach is to compare the number of edges within the same cluster to its expected value in a random graph that inherits some properties of the graph given as input.

In a wider sense, the measure called~\emph{surprise} that has recently been suggested as an alternative to modularity is also based on a null model, although, compared to modularity and its modifications~\cite{f-c-09}, it uses a different tradeoff between the observed and expected number of edges within clusters.
Surprise is used as a quality function in the tools UVCLUSTER and Jerarca to analyze protein interaction data~\cite{amm-icapi-05,am-jeacn-10}.
The authors' main arguments for using surprise instead of modularity is that it exhibits better behavior with respect to synthetic benchmarks and, empirically, it does not suffer to the same extent from the \emph{resolution limit} of modularity~\cite{bf-rlcd-07}, i.e.~the tendency to merge small natural communities into larger ones~\cite{am-dncss-11,am-e-13,am-s-13}.
However,  these results are hard to assess, since a metaheuristic is used instead of directly optimizing the measure. 
It chooses among a set of clusterings produced by general clustering algorithms the one that is best with respect to surprise. 

In this work, we take first steps towards a theoretical analysis of surprise.
We show that the problem of finding a clustering with optimal surprise is $\mathcal{NP}$-hard in general and polynomially solvable on trees. 
Moreover, we formulate surprise as a bicriterial problem, which allows to find provably optimal solutions for small instances by solving a small number of integer linear programs.

\vspace{2ex}
\noindent
\textbf{Notation.} 
All graphs considered are 
unweighted, undirected and simple, i.e.~they do not contain loops or parallel edges. 
A clustering $\C$ of a graph $G=(V,E)$ is a partitioning of $V$.
Let $n:=|V|$ and $m:=|E|$ denote the number of vertices and edges of $G$, respectively.
If $C$ is a cluster in $\C$, $i_e(C)$ denotes the number of \emph{intracluster edges} in $C$, i.e., the number of edges having both endpoints in $C$.
Similarly, $i_p(C) := \binom{|C|}{2}$ is the number of vertex pairs in $C$.
Furthermore, let $p := \binom{n}{2}$ be the number of vertex pairs in $G$, $i_p(\C) := \sum_{C \in \C} i_p(C)$ be the total number of intracluster vertex pairs and $i_e(\C) := \sum_{C \in \C} i_e(C)$ the total number of intracluster edges.
If the clustering is clear from the context, we will sometimes omit $\C$ and just write $i_p$ and $i_e$.   
To ease notation, we will allow binomial coefficients $\binom{n}{k}$ for all $n$ and $k \in \mathbb{N}$.
If $k > n$, $\binom{n}{k} = 0$ by definition. 

\section{Definition and Basic Properties}
Let $\C$ be a clustering of a graph $G=(V,E)$ with $i_e$ intracluster edges.
Among all graphs labeled with vertex set $V$ and exactly $m$ edges, we draw a graph $\mathcal{G}$ uniformly at random. 
The surprise $S(\C)$ of this clustering is then the probability that $\mathcal{G}$ has at least $i_e$ intracluster edges with respect to $\C$.
The lower this probability, the more \emph{surprising} it is to observe that many intracluster edges within $G$, and hence, the better the clustering.
The above process corresponds to an urn model with $i_p(\C)$ white and $p-i_p(\C)$ black balls from which we draw $m$ balls without replacement.
The probability to draw at least $i_e$ white balls then follows a hypergeometric distribution, which leads to the following definition\footnote{This is the definition used in the original version~\cite{amm-icapi-05}; later on, it was replaced by maximizing $-\log_{10} S(\C)$, which is equivalent with respect to optimum solutions.}; the lower $S(\C)$, the better the clustering:
$$S(\C) := \sum_{i=i_e}^{m} \frac{\binom{i_p}{i} \cdot \binom{p - i_p}{m-i}}{\binom{p}{m}}$$

\noindent
\textbf{Basic Properties.} 
For a fixed graph, the value of $S$ only depends on two variables, $i_p$ and $i_e$.
To ease notation, we will use the term $S(i_p, i_e)$ for the value of a clustering with $i_p$ intracluster pairs and $i_e$ intracluster edges.
The urn model view 
yields some simple properties 
that lead to a better understanding of how surprise behaves, and that are heavily used in the $\mathcal{NP}$-hardness proof.

\begin{lemma}
Let $i_e$, $i_p$, $p$ and $m$ be given by a clustering, i.e.~$0\leq i_e \leq i_p \leq p$, $i_e \leq m$ and $m-i_e \leq p-i_p$.
Then, the following statements hold:
\begin{enumerate}
 \item[(i)] $S(i_p, i_e+1) < S(i_p, i_e).$
 \item[(ii)] If $i_e > 0$, then $S(i_p-1, i_e) < S(i_p, i_e)$
 \item[(iii)] If $p - i_p > m - i_e$, then $S(i_p+1, i_e+1) < S(i_p, i_e).$
\end{enumerate}
  \label{lem:basic_props}
\end{lemma}
\begin{proof}
 Statement (i) is obvious.
 Similarly, statement (ii) is not hard to see if we recall that $S(i_p-1, i_e)$ corresponds to the probability to draw at least $i_e$ white balls after replacing one white ball with a black one.
 
 For statement (iii), we show that the number $k_1$ of $m$-element subsets of the set of all balls containing at least $i_e$ white balls is larger than the number $k_2$ of $m$-element subsets containing at least $i_e+1$ white balls after painting one black ball $b$ white.
 Any subset $A$ that contributes to $k_2$ also contributes to $k_1$, as at most one ball in $A$ got painted white.
 On the other hand, every $m$-element subset not containing $b$ that contains exactly $i_e$ white balls contributes to $k_1$, but not to $k_2$.
 As there are at least $i_e$ white balls, and $p-i_p > m-i_e$ implies that there are at least $m-i_e+1$ black balls, there is at least one subset with these properties.
 Hence $k_1 > k_2$, which is equivalent to $S(i_p+1, i_e+1) < S(i_p, i_e)$.\qed
\end{proof}

In other words, the value of surprise improves the more edges and the less vertex pairs within clusters exist.
Moreover, part (iii) shows that if we increase the number of intracluster edges such that the number of \emph{intracluster non-edges}, i.e., vertex pairs within clusters that are not linked by an edge, does not increase, this leads to a clustering with strictly smaller surprise.
This immediately yields some basic properties of optimal clusterings with respect to surprise.
Part (i) of the following proposition is interesting as it shows that optimal clusterings always fulfill the assumptions of Lemma~\ref{lem:basic_props}(ii)-(iii).
\begin{proposition}
 Let $G=(V,E)$ be a graph that has at least one edge and that is not a clique and $\C$ be an optimal clustering of $G$ with respect to surprise. Then,
 \begin{enumerate}
  \item[(i)] $i_e(\C) > 0$ and $p-i_p(\C) > m-i_e(\C)$
  \item[(ii)] $1 < |\C| < |V|$
  \item[(iii)] $\C$ contains at least as many intracluster edges as any clustering $\C'$ of $G$ into cliques.
  \item[(iv)] Any cluster in $\C$ induces a connected subgraph.
 \end{enumerate}
 \label{lem:props_opt_sol}
\end{proposition}
\begin{proof}
(i): If $i_e(\C) = 0$ or $p - i_p(\C) = m - i_e(\C)$, it can be easily seen that $S(\C)=1$.
On the other hand, let us consider a clustering $\C'$ where each cluster contains one vertex, except for one cluster that contains two vertices linked by an edge $e$.
As $m<p$, there is at least one labeled graph on $V$ with $m$ edges that does not contain $e$.

(ii): 
If $|\C|=1$, $p-i_p(\C) = 0 = m-i_e(\C)$ and if $|\C|=|V|$, $i_e(\C)=0$.
The statement now follows from (i).

(iii): Let us assume that $i_e(\C) < i_e(\C')$.
Lemma~\ref{lem:basic_props}(ii) can be used to show that $S(\C) = S\bigl(i_p(\C), i_e(\C)\bigr) \geq S\bigl(i_e(\C), i_e(\C)\bigr)$ and from Lemma~\ref{lem:basic_props}(iii), it follows that $S\bigl(i_e(\C), i_e(\C)\bigr) > S\bigl(i_e(\C'), i_e(\C')\bigr) = S(\C')$.

(iv): Follows from Lemma~\ref{lem:basic_props}(ii) and the fact that splitting a disconnected cluster into its connected components decreases the number of intracluster pairs and does not affect the number of intracluster edges. \qed
\end{proof}

\vspace{2ex}
\noindent
\textbf{Bicriterial View.} 
\label{sec:connection}
From Lemma~\ref{lem:basic_props}, it follows that an optimal solution with respect to surprise is \emph{pareto optimal} with respect to (maximizing) $i_e$ and (minimizing) $i_p$.
Interestingly, this also holds for a simplification of modularity whose null model does not take vertex degrees into account and that was briefly considered by Reichardt and Bornholdt~\cite{rb-smcd-06,rb-dfcsc-04}, although the tradeoff between the two objectives is different. 
Hence, an optimal clustering 
can be found by solving the following optimization problem for all $0 \leq k \leq m$ and choosing the solution that optimizes surprise.
\begin{problem}[\mipShort]
 Given a graph $G$ and an integer $k>0$, find a clustering $\C$ with $i_e(\C) = k$, if there exists one, such that $i_p(\C)$ is minimal.
\end{problem}
Unfortunately, the decision variant of \mipShort~is $\mathcal{NP}$-complete even on bipartite graphs, as it is equivalent to the unweighted Minimum Average Contamination problem~\cite{lt-tcamc-11}.
However, the formulation of \mipShort~does not involve binomial coefficients and is thus in some aspects easier to handle.
For example, in contrast to surprise, it can be easily cast into an integer linear program.
We will use this in Sect.~\ref{sec:algorithms} to compute optimal solutions for small instances.

One might guess from the $\mathcal{NP}$-completeness of \mipShort~that surprise minimization is also $\mathcal{NP}$-complete.
However, there is no immediate reduction from \mipShort~to the decision variant of surprise optimization, as the number of intracluster edges in an optimal clustering with respect to surprise is not fixed.
In the following section, we will therefore give a proof for the hardness of finding a clustering with optimal surprise.

\section{Complexity}\label{sec:complexity}
We show $\mathcal{NP}$-completeness of the corresponding decision problem:
\begin{problem}[\textsc{Surprise Decision (SD)}]
 Given a graph $G$ and a parameter $k>0$, decide whether there exists a clustering $\C$ of $G$ with $S(\C) \leq k$.
\end{problem}
As $S$ can be clearly evaluated in polynomial time, \textsc{SD}~is in $\mathcal{NP}$.
To show $\mathcal{NP}$-completeness, we use a reduction from \textsc{Exact Cover by 3-Sets}~\cite{gj-ci-79}:
\begin{problem}[\textsc{Exact Cover by 3Sets (X3S)}]
 Given a set $\X$ of elements and a collection $\Set$ of 3-element subsets of $\X$, 
 decide whether there is a subcollection $\R$ of $\Set$ such that each element in $\X$ is contained in exactly one member of $\R$.
\end{problem}

\begin{wrapfigure}[10]{r}{.4\textwidth}
\vspace{-8ex}
 \begin{center}
   \includegraphics[scale = 0.7]{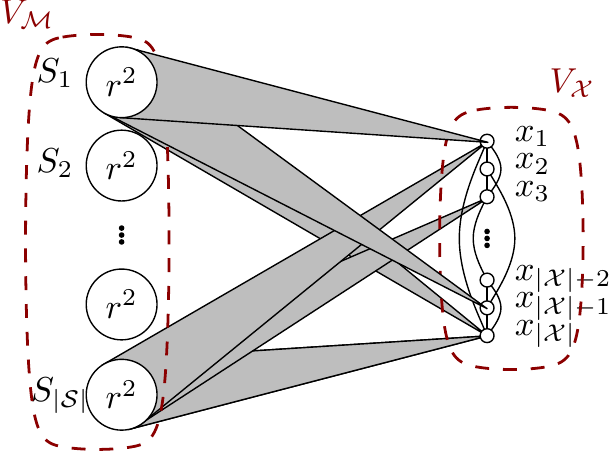}
 \end{center}
 \vspace{-2ex}
 \caption{Illustration for reduction.}
 \label{fig:reduction}
\end{wrapfigure}
Let $I = (\X, \Set)$ be an instance of \textsc{X3S}.
The reduction is based on the idea of implanting large disjoint cliques in the transformed instance that correspond to the subsets in $\Set$. The size of these cliques is polynomial in $|\Set|$, but large enough to ensure that they can neither be split nor merged in a clustering with low surprise.
Hence, each of these cliques induces a cluster.
The transformed instance further contains a vertex for each element in $\X$ that is linked with the cliques corresponding to the subsets it is contained in.
The idea is to show that in a clustering $\C$ with low surprise, each of these vertices is contained in a cluster induced by exactly one subset, and each cluster contains either three ``element vertices'' or none, which induces an exact cover of $\X$.

In the following, we will assume without loss of generality\footnote{Otherwise, the instance is trivially non-solvable.} that each element of $\X$ belongs to at least one set in $\Set$, hence $|\X| \leq 3|\Set|$.
We construct an instance $I' = (G,k)$ of \textsc{SD} in the following way.
Let $r := 3|\Set|$.
First, we map each set $M$ in $\Set$ to an $r^2$-clique $C(M)$ in $G$.
Furthermore, we introduce an $|\X|$-clique to $G$, where each of the vertices $v(x)$ in it is associated with an element $x$ in $\X$.
We link $v(x)$ with each vertex in $C(M)$, if and only if $x$ is contained in $M$.
Let $V_\X$ be the set containing all vertices corresponding to elements in $\X$, and $V_\Set$ the set of vertices corresponding to subsets.
Fig.~\ref{fig:reduction}~illustrates the reduction, clearly, it is polynomial.
In the proof, we will frequently use the notion \emph{for large $r$, statement $A(r)$ holds}.
Formally, this is an abbreviation for the statement that there exists a constant $c>0$ such that for all $r \geq c$, $A(r)$ is true.
Consequently, the reduction only works for instances that are larger than the maximum of all these constants, which suffices to show that \textsc{SD} is $\mathcal{NP}$-complete\footnote{Smaller instances have constant size and can therefore be trivially solved by a brute-force algorithm.}.
 
\begin{lemma}
 \label{lem:os_many_ie}
Let $\C$ be an optimal clustering of $G$ with respect to $S$. Then, $i_e(\C) \geq |\Set| \cdot \binom{r^2}{2}$.
\end{lemma}
\begin{proof}
Follows from Proposition~\ref{lem:props_opt_sol}(iii) and the fact that the clustering whose clusters are the cliques in $V_\Set$ and the singletons in $V_\X$ is a clustering into cliques with $|\Set| \cdot \binom{r^2}{2}$ intracluster edges. \qed
\end{proof}
Next, we give an upper bound on the number of \emph{intracluster non edges}, i.e., vertex pairs within clusters that are not linked by an edge, in an optimal clustering of $G$.
Its (rather technical) proof makes use of the asymptotic behavior of binomial coefficients and can be found in App.~\ref{app:asymptotics}.

\begin{lemma}
 Let $\C$ be an optimal clustering of $G$ with respect to surprise. Then, for large $r$, $i_p(\C) - i_e(\C) \leq \frac{r^4}{2}.$
 \label{lem:os_few_intra_non}
\end{lemma}

This can now be used to show that an optimal clustering of $G$ is a clustering into cliques.
We start by showing that the cliques in $V_\Set$ cannot be split by an optimal clustering.
\begin{figure}[tbp]
  \begin{center}
   \includegraphics[scale = 1]{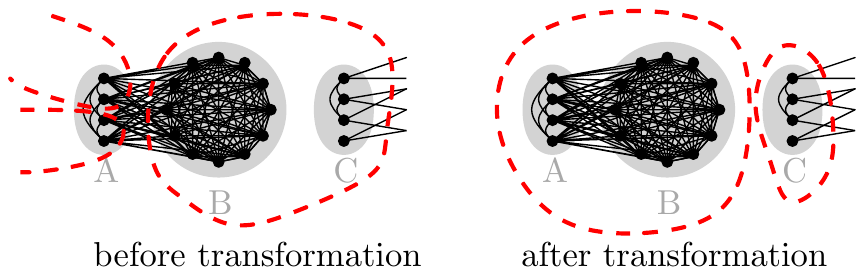}
 \end{center}
 \caption{Illustration for proof of Lemma~\ref{lem:clique_not_split}}
 \label{fig:Illu_Proof}
\end{figure}
  
  \begin{lemma}
 \label{lem:clique_not_split}
Let $r$ be large and $\C$ be an optimal clustering of $G$ with respect to $S$.
Then, the cliques $C(M)$ in $V_\Set$ are not split by $\C$.
\end{lemma}
\begin{proof}
Assume that there is at least one clique that is split by $\C$.
$\C$ induces a partition of each clique that it splits.
We call the subsets of this partition the \emph{parts} of the clique.

\emph{Claim 1: Every clique $C(M)$ contains a part with at least $r^2-6$ vertices.}

\emph{Proof of Claim 1:}
Assume that there is a clique $K$ where each part has at most $r^2-7$ vertices.
We can now greedily group the parts in two roughly equal sized regions, such that the smaller region contains at least $7$ vertices and the larger region at least $r^2/2$ vertices.
Let us look at the clustering we get by removing the vertices in $K$ from their clusters and cluster them together.
The vertices in $K$ have in total $3r^2$ edges to vertices outside $K$ and we gain at least $7/2 \cdot r^2$ new intracluster edges between the regions.
Hence, the number of intracluster edges increases and the number of intracluster non-edges can only decrease.
By Lemma~\ref{lem:basic_props}(iii) and Lemma~\ref{lem:basic_props}(i), it can be seen that this operation leads to a clustering with better surprise, which contradicts the optimality of $\C$.

Let us now call the parts with size at least $r^2-6$ \emph{large parts} and the other parts \emph{small parts}.

\emph{Claim 2: No two large parts are clustered together.}

\emph{Proof of Claim 2:}
Assume that there is a cluster that contains more than one large part.
This cluster induces at least $(r^2-6)^2$ intracluster non-edges.
For large $r$, this is larger than $r^4/2$ and Lemma~\ref{lem:os_few_intra_non} tells us that $\C$ was not optimal.

A simple counting argument now yields the following corollary.

\emph{Corollary: There must exist a large part $B$ contained in a split clique whose cluster contains at most $|B|+6$ vertices in $V_\Set$}.

Let $B$ as in the corollary and $A$ be the set of the vertices that are in the same clique as $B$ but not in $B$ and $C$ be the set of vertices that are in the same cluster as $B$ but not in $B$.
Fig.~\ref{fig:Illu_Proof} illustrates this case.
We consider the clustering that we get by removing the vertices in $A$ and $B$ from their cluster and cluster them together.
The number of vertices in $A$ and $C$, respectively, is at most $6$, and each of these vertices has at most $3$ neighbors in $V_\X$.
Hence, we lose at most $36$ intracluster edges by this operation.
On the other hand, we gain at least $r^2-6$ intracluster edges between $A$ and $B$, thus, for large $r$, the number of intracluster edges increases.
Again, the number of intracluster non-edges can only decrease and by Lemma~\ref{lem:basic_props}(iii) and Lemma~\ref{lem:basic_props}(i), we get that this operation leads to a clustering with better surprise, which contradicts the optimality of $\C$. \qed 
\end{proof}

\begin{lemma}
\label{lem:cliques_apart}
 Let $r$ be large and $\C$ be an optimal clustering of $G$ with respect to $S$.
 Then, no two of the cliques in $V_\Set$ are contained in the same cluster.
\end{lemma}
\begin{proof}
 A cluster that contains two cliques in $V_\Set$ induces at least $r^4$ intracluster non-edges.
 The statement now follows from Lemma~\ref{lem:os_few_intra_non}. \qed
\end{proof}

\begin{lemma}
 \label{lem:vertices_assigned}
Let $r$ be large and $\C$ an optimal clustering of $G$ with respect to $S$.
Then, each $v(x)$ in $V_\X$ shares a cluster with a clique $C(M)$ such that $x \in M$.
\end{lemma}
\begin{proof}From Lemma~\ref{lem:clique_not_split} and Lemma~\ref{lem:cliques_apart} we know that $\C$ clusters the vertices in $V_\Set$ according to the cliques we constructed.
 Assume that there is a vertex $v(x)$ in $V_\X$ that is not contained in any of the clusters induced by the sets containing $x$.
  Since each element in $\X$ is contained in at least one set in $\Set$,  there exists a clique $K$ in $V_\Set$ that contains $r^2$ neighbors of $v(x)$.
  As $v(x)$ has at most $|\X| -1$ neighbors in its own cluster, removing it from its cluster and moving it to the cluster of $K$ increases the number of intracluster edges.
  On the other hand, $x$ is linked with all vertices in its new cluster and thus, the number of intracluster non-edges cannot increase.
  Hence, this operation leads to a clustering with better surprise, which contradicts the optimality of $\C$. \qed
\end{proof}

\begin{theorem}
 For large $r$, $I=(\X,\Set)$ has a solution if and only if there exists a clustering $\C$ of $G$ with $S(\C) \leq k := {\binom{p}{m}}^{-1} \cdot \left({\binom{|\Set| \cdot r^2 + |\X|}{2} - |\Set| \cdot \binom{r^2}{2} - |\X| \cdot r^2 - |\X| \atop  (3|\Set| - |\X|) \cdot r^2 + \binom{|\X|}{2} - |\X|} \right)$.
\end{theorem}
\begin{proof}
  $\Rightarrow$: Let $R$ be a solution of $I$. 
  $R$ induces a clustering of $G$ in the following way: For each $M \in \Set \setminus R$ we introduce a cluster $C_M=C(M)$ and for each $M' \in R$ a cluster $C_{M'} = C(M') \cup \{v(x) \mid x \in M'\}$.
  As $R$ is an exact cover, this is a partition $\C$ of the vertex set.
  It is $p=\binom{|\Set| \cdot r^2 + |\X|}{2}$, $m=|\Set| \cdot \binom{r^2}{2} + 3 \cdot |\Set| \cdot r^2 + \binom{|\X}{2}$ and $i_p(\C)=i_e(\C)=|\Set| \cdot \binom{r^2}{2}  + |\X| \cdot r^2 + |\X|$.
  It can be easily verified that $S(\C) = k$.
  
  $\Leftarrow$: Let $\C$ be an optimal clustering of $G$ with respect to surprise and assume that $S(\C) \leq k$.
  From Lemma~\ref{lem:clique_not_split}, Lemma~\ref{lem:cliques_apart} and Lemma~\ref{lem:vertices_assigned}, we know that, for large $r$, we have one cluster for each set $M$ in $\Set$ that contains $C(M)$ and each vertex $v(x)$ in $V_\X$ shares a cluster with a clique $C(M)$ such that $x \in M$.
  In particular, all clusters in $\C$ are cliques and hence $\binom{i_p(\C)}{i_e(\C)} = 1$.
  It follows that $\binom{p}{m} \cdot k \geq \binom{p}{m} \cdot S(\C) =  \binom{p - i_e(\C)}{m - i_e(\C)}$.
  This term is strictly decreasing with $i_e(\C)$ and the above bound is tight for $i_e(\C) = |\Set| \cdot \binom{ \cdot r^2}{2} + |\X| \cdot r^2 + |\X|:= t$.
  Hence, $\C$ contains at least $t$ intracluster edges.
  The number of intracluster edges within $V_\Set$ is exactly $|\Set| \cdot \binom{r^2}{2}$ and the number of intracluster edges linking $V_\Set$ with $V_\X$ is exactly $|\X| \cdot r^2$.
  The only quantity we do not know is the number of intracluster edges within $V_\X$, which we denote by $i_e(V_\X)$.
  As $i_e(\C) \geq t$, it follows that $i_e(V_\X) \geq |\X|$.
  Thus, every vertex in $V_\X$ has in average two neighbors in $V_\X$ that are in the same cluster.
  On the other hand, vertices in $V_\X$ can only share a cluster if they are ``assigned'' to the same clique $C(M)$.
  As the sets in $\Set$ only contain three elements, vertices in $V_\X$ can only have at most two neighbors in $V_\X$ in their cluster.
  It follows that $\C$ partitions $V_\X$ into triangles.
  Hence, the set of subsets $R$ corresponding to cliques $C(M)$ whose clusters contain vertices in $V_\X$ form an exact cover of $\X$. \qed
\end{proof}

We now have a reduction from \textsc{X3S} to \textsc{SD} that works for all instances that are larger than a constant $c > 0$.
Hence, we get the following corollary.
\begin{corollary}
 \textsc{Surprise Decision} is $\mathcal{NP}$-complete.
\end{corollary}

To show that an optimal clustering with respect to surprise can be found in polynomial time if $G$ is a tree, we consider the following problem MACP~\cite{lt-tcamc-11}:
\begin{problem}[MACP]
Given a graph $G=(V,E)$ together with a weight function $w:V\rightarrow \mathbb{Q}_{\geq 0}$ on $V$ and a parameter $k$.
Find a clustering $\C$ of $G$ such that $m - i_e(\C)=k$ and $\sum_{C \in \C} {\bigl(\sum_{v \in C} w(v)\bigr)}^2$ is minimal.
\end{problem}
For the special case that $w(v)$ equals the degree of $v$ and $G$ is a tree, Dinh and Thai give a dynamic program that solves MACP for all $0 \leq k\leq m$ simultaneously~\cite{dt-tocdf-}.
This yields an $O(n^5)$ algorithm for modularity maximization in (unweighted) trees.
In the context of surprise, we are interested in the special case that $w(v)=1$ for all $v \in V$.
 The following conversion shows that this is equivalent to \mipShort~with respect to optimal solutions:
\begin{equation} \label{minIP_MACP}
i_p(\mathcal{C}) 
= \sum_{C \in \mathcal{C}} {\frac{\left|C\right|(\left|C\right| - 1)}{2}}
= \frac{1}{2} \sum_{C \in \mathcal{C}} {\left|C\right|^2} -
\underbrace{\frac{1}{2} \left|V\right|}_{=\mathrm{const.}}
\end{equation}

The dynamic program of Dinh and Thai has a straightforward generalization to general vertex weights, which is polynomial in the case that each vertex has weight $1$.
For completeness, App.~\ref{app:trees} contains a description of the dynamic program in this special case, together with a runtime analysis.
\begin{theorem}\label{theoremTrees}
Let $T = (V,E)$ with $n:=\left| V \right|$ be an unweighted tree. Then, a surprise optimal clustering of $T$ can be calculated in $O(n^5)$ time.
\end{theorem}
\section{Exact Solutions}
\label{sec:algorithms} 
  In this section, we give an integer linear program for \mipShort~and discuss some variants of how to use this to get optimal clusterings with respect to surprise.
  
  \vspace{2ex}
  \noindent
  \textbf{Linear Program for \mipShort.} 
  The following ILP is very similar to a number of linear programs used for other objectives in the context of graph clustering and partitioning, in particular, to one used for modularity maximization~\cite{dt-tocdf-}.
  It uses a set of $\binom{n}{2}$ binary variables $\Xer{u}{v}$ corresponding to vertex pairs, with the interpretation that $\Xer{u}{v}=1$ iff $u$ and $v$ are in the same cluster.
  Let $\mathrm{Sep}(u,v)$ be a minimum $u$-$v$ vertex separator in $G$ if $\{u,v\} \notin E$ or in $G'=(V, E\setminus{\{u,v\}})$, otherwise.
  The objective is to
  \begin{align}
    \text{minimize }\sum_{\{u,v\} \in \binom{V}{2}} \Xer{u}{v} & \label{eq:basic_objective}
 \end{align}
 such that
 \begin{align}
 \renewcommand{\baselinestretch}{1.7}\normalsize
    \Xer{u}{v} \in \{0,1\}, &\quad \{u,v\} \in \binom{V}{2} \\
      \Xer{u}{w} + \Xer{w}{v} - \Xer{u}{v} \leq 1,  &\quad \{u,v\} \in \binom{V}{2}, w \in \mathrm{Sep}(u,v) \label{con:transitivity}\\
    \sum_{\{u,v\} \in E} \Xer{u}{v} = k &\quad \label{con:exactly_k}
 \end{align}
 Dinh and Thai consider the symmetric and reflexive relation induced by $\mathcal{X}$ and show that Constraint~(\ref{con:transitivity}) suffices to enforce transitivity in the context of modularity maximization~\cite{dt-tocdf-}.
 Their proof solely relies on the following argument.
 For an assignment of the variables $\Xer{u}{v}$ that does not violate any constraints, let us consider the graph $G'$ induced by the vertex pairs $\{u,v\}$ with $\Xer{u}{v}=1$.
 Now assume that there exists a connected component in $G'$ that can be partitioned into two subsets $A$ and $B$ such that there are no edges in the original graph $G$ between them.
 Setting $\Xer{a}{b} := 0$ for all $a \in A$, $b \in B$ never violates any constraints and strictly improves the objective function.
 It can be verified that this argument also works in our scenario.
 Hence, a solution of the above ILP induces an equivalence relation and therefore a partition of the vertex set.
 As $\mathrm{Sep}(u,v)$ is not larger than the minimum of the degrees of $u$ and $v$, we have $O(nm)$ constraints over $O(n^2)$ variables.
  
  \vspace{2ex}
  \noindent
  \textbf{Variants.} 
  We tested several variants of the approach described in Sect.~\ref{sec:introduction} to decrease the number of ILPs we have to solve.
  \begin{itemize}
   \item \emph{Exact(E)}: Solve $m$ times the above ILP and choose among the resulting clusterings the one optimizing surprise.
   \item \emph{Relaxed(R)}:   We relax Constraint~(\ref{con:exactly_k}), more specifically we replace it by
    \begin{equation}
    \sum_{\{u,v\} \in E} \Xer{u}{v} \geq k
    \label{con:at_least_k}
  \end{equation}
  Lemma~\ref{lem:basic_props}(i) tells us that the surprise of the resulting clustering is at least as good as the surprise of any clustering with exactly $k$ intracluster edges.
  Moreover, by Lemma~\ref{lem:basic_props}(ii), if $i_p$ is the value of a solution to the modified ILP, $S(i_p, k')$ is a valid lower bound for the surprise of any clustering with $k'\geq k$ intracluster edges.
  In order to profit from this, we consider all possible values for the number of intracluster edges in increasing order and only solve an ILP if the lower bound is better than the best solution found so far.
    \item \emph{Gap(G)}:    Similarly to the relaxed variant, we replace Constraint~(\ref{con:exactly_k}) by~(\ref{con:at_least_k}) and modify~(\ref{eq:basic_objective}) to
    \begin{equation}
    \text{minimize }\sum_{\{u,v\} \in \binom{V}{2}} \Xer{u}{v} - \sum_{\{u,v\} \in E} \Xer{u}{v}
    \label{eq:gap_objective}
    \end{equation}
    By Lemma~\ref{lem:basic_props}(ii), if $g$ is the objective value and $i_e$ the number of intracluster edges in a solution to the modified ILP, $S(k'+g, k')$ is a valid lower bound for the surprise of any clustering with $k'\geq k$ intracluster edges.
    Moreover, by Lemma~\ref{lem:basic_props}(iii), we know that $S(i_e+g, i_e)$ is not larger than the surprise of any clustering with exactly $k$ intracluster edges.
    Again, we consider all $k$ in increasing order and try to prune ILP computations with the lower bound.
  \end{itemize}

\vspace{2ex}
\noindent
\textbf{Case Study.} 
     \begin{table}[btp]
   \caption{Number of linear programs solved and running times in seconds of successive ILP approach, different strategies.}
   \label{tab:running_times}
   \begin{center}
   \begin{footnotesize}
\begin{tabular}{|l|cc|cc|cc|cc|} 
\hline \multicolumn{1}{|c|}{} & \multicolumn{2}{c|}{\texttt{karate}} & \multicolumn{2}{c|}{\texttt{lesmis}} & \multicolumn{2}{c|}{\texttt{grid6}} & \multicolumn{2}{c|}{\texttt{dolphins}} \\ 
variant & ILP & t(s) & ILP & t(s) & ILP & t(s) & ILP & t(s) \\ 
\hline Exact & 79 & 51 & 255 & 1192 & 61 & 470 & 160 & 494 \\ 
Relaxed & 49 & 21 & 176 & 282 & 42 & 449 & 107 & 163 \\
Gap & 39 & \textbf{15} & 112 & \textbf{205} & 37 & \textbf{401} & 91 & \textbf{147} \\ 
\hline \end{tabular} 
  \end{footnotesize}
  \end{center}
  \end{table}
  Table~\ref{tab:running_times}~shows an overview of running times and the number of solved ILPs of the different strategies on some small instances.
  \texttt{karate}($n=34, m=78$), \texttt{dolphins}($n=62,m=159$) and \texttt{lesmis}($n=77,m=254$) are real world networks from the website of the 10th DIMACS implementation Challenge\footnote{\url{http://www.cc.gatech.edu/dimacs10/}} that have been previously used to evaluate and compare clusterings, whereas 
  \texttt{grid6}($n=36, m=60$) is a 2 dimensional grid graph.  
  We used the C++-interface of \texttt{gurobi5.1}~\cite{gurobi} and computed the surprise of the resulting clusterings with the help of the GNU Multiple Precision Arithmetic Library, in order to guarantee optimality.
  The tests were executed on one core of an AMD Opteron Processor 2218.
  The machine is clocked at 2.1 GHz and has 16 GB of RAM.
  Running times are averaged over $5$ runs.

It can be seen that the gap variant, and, to a smaller extent, the relaxed variant, are able to prune a large percentage of ILP computations and thus lead to less overall running time.
These running times can be slightly improved by using some heuristic modifications described and evaluated in App.~\ref{app:heuristics}.
%

\vspace{2ex}
\noindent
\textbf{Properties of optimal clusterings.}
\begin{figure}[tbp]
 \begin{center}
  \subfigure[\texttt{karate}]{\includegraphics[scale=0.2]{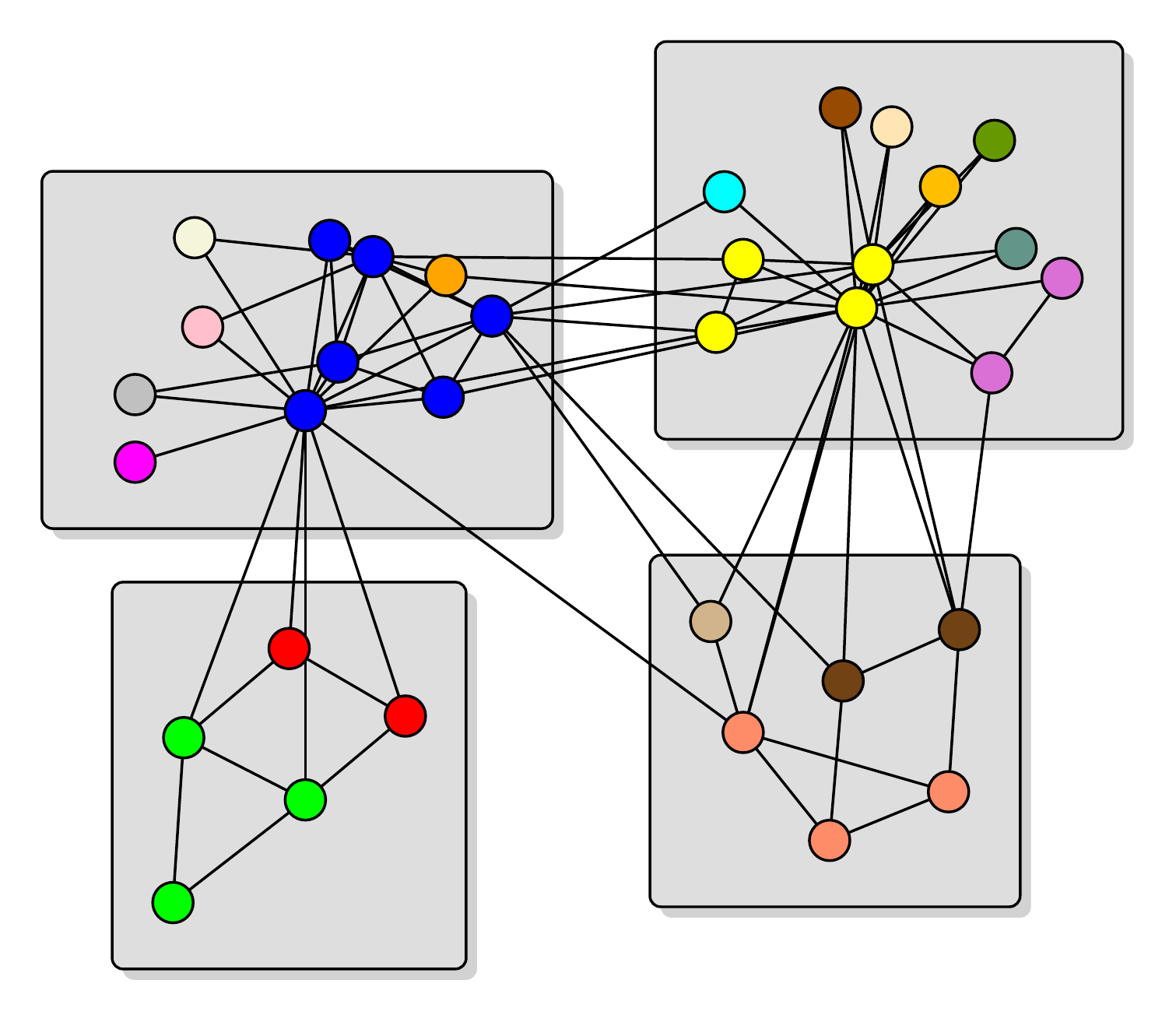}}
  \subfigure[\texttt{dolphins}]{\includegraphics[scale=0.2]{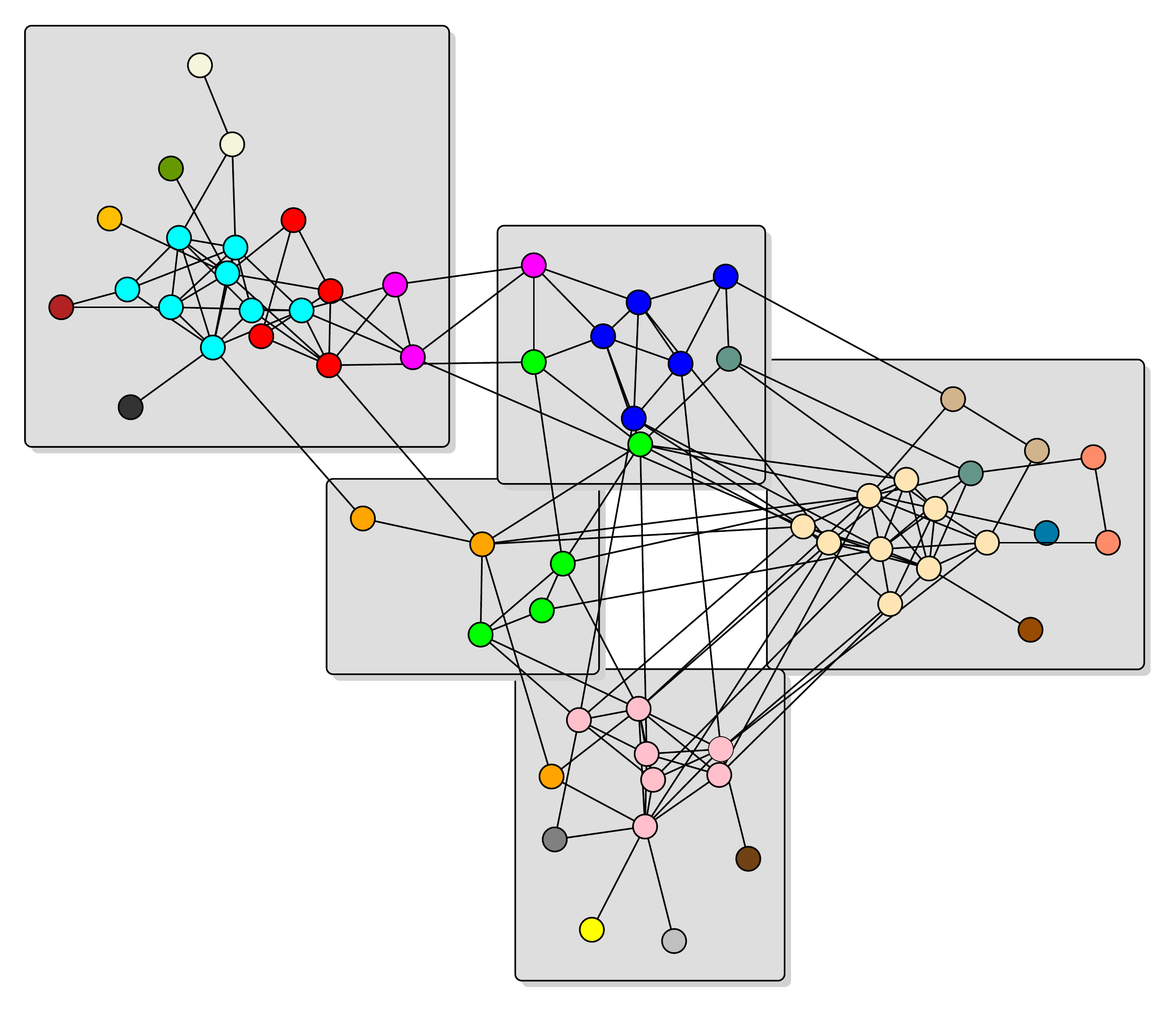}}
  \subfigure[\texttt{grid6}]{\includegraphics[scale=0.2]{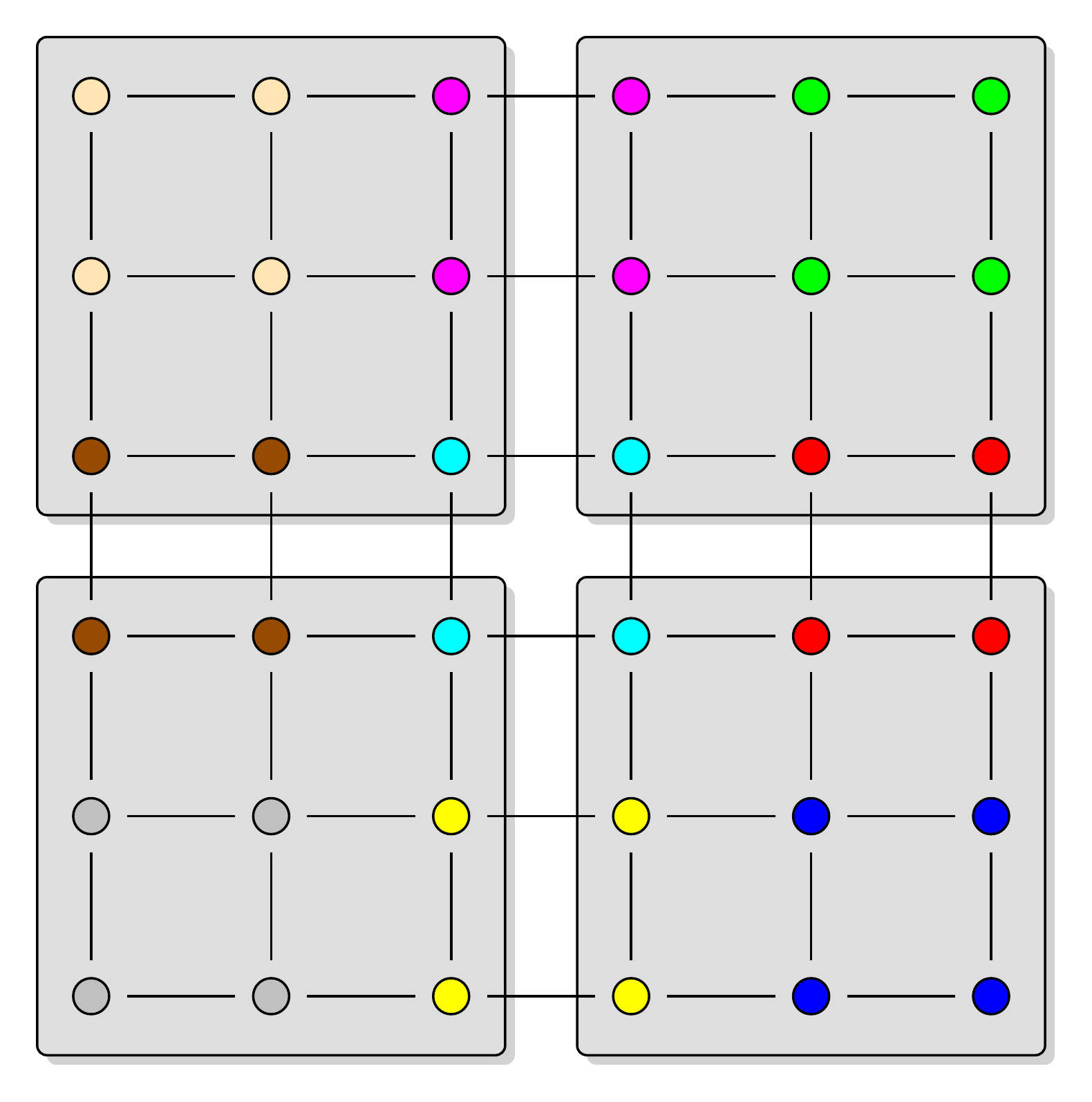}}
  \subfigure[\texttt{lesmis}]{\includegraphics[scale=0.35]{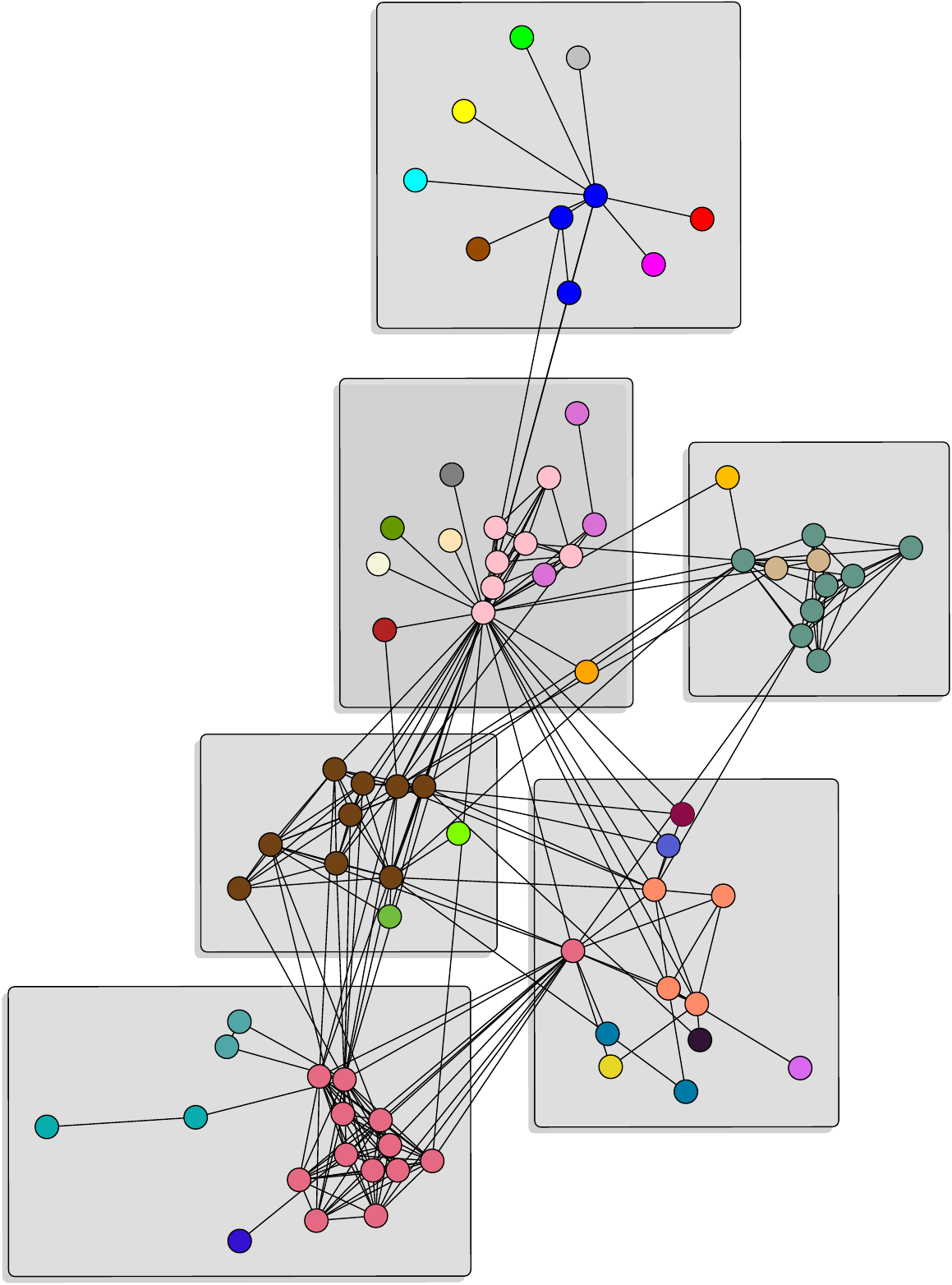}}
    \subfigure[\texttt{football}]{\includegraphics[scale=0.75]{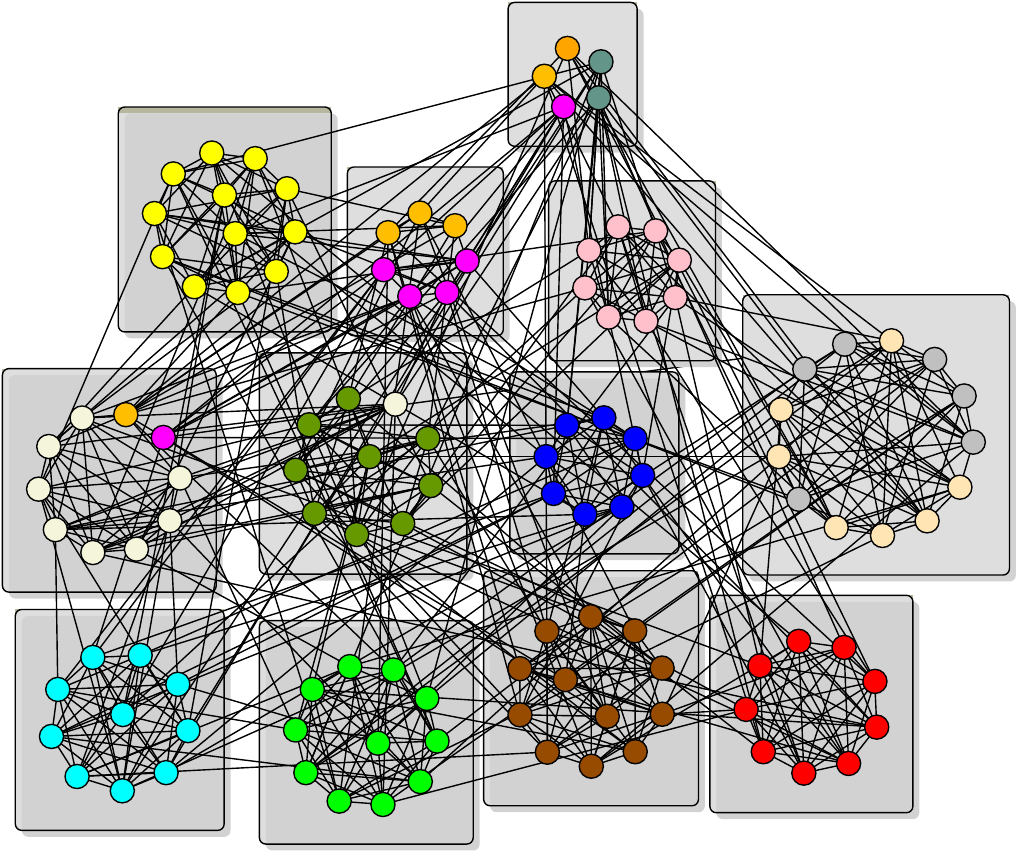}}
  \caption{Optimal clusterings with respect to surprise(colors) and, for (a) to (d), modularity(grouping). The grouping in (e) represents the \emph{ground-truth clustering}, i.e.~the mapping of teams to conferences.}
  \label{fig:opt_sol}
 \end{center}
\end{figure}
\begin{table}[p]
\caption{Properties of optimal clusterings with respect to surprise. $S'$ denotes the surprise as defined by Aldecoa and Mar{\'\i}n~\cite{am-dncss-11}, i.e. $S'(\C) = -\log_{10} S(\C)$. $S_o$ denotes the clustering with optimum surprise, $S_h$ the heuristically found clusterings from~\cite{am-dncss-11}, if this information was available, and $M_o$ the modularity optimal clustering.}
\label{tab:prop_opt}
\nprounddigits{2}
\begin{center}
\begin{tabular}{|l|r|r|l|r|r|r|r|r|}
\hline instance & $i_e$ & $i_p$ & $S(S_o)$ & $S'(S_o)$ &$S'(S_h)$  & $|S_o|$ & $|S_h|$ & $|M_o|$\\
\hline
\texttt{karate} & $29$ & $30$ & $\numprint{2.02474e-26}$ & $25.69$ & $25.69$ & $19$ & $19$ & $4$\\
\texttt{grid6} & $36$ & $54$ & $\numprint{2.89981e-29}$ & $28.54$ & - & $9$ & - & $4$ \\
\texttt{dolphins} & $87$ & $121$ & $\numprint{9.93152e-77}$ & $76.00$ & - & $22$ & - & $5$\\
\texttt{lesmis} & $165$ & $179$ & $\numprint{1.5385e-184}$ & $183.81$ & -  &$33$ & - & $6$\\
\texttt{football} &  $399$ & $458$ & $\numprint{5.64724e-407}$ &  $\numprint{406.248}$& -& $15$ & $15$ & $10$ \\ \hline
\end{tabular}
\end{center}
\end{table}
Fig.~\ref{fig:opt_sol} illustrates optimal clusterings with respect to surprise and modularity on the test instances, Table~\ref{tab:prop_opt} summarizes some of their properties.
We also included one slightly larger graph, \texttt{football}($n=115,m=613$), as it has a known, well-motivated \emph{ground truth clustering} and has been evaluated in~\cite{am-dncss-11}.
The surprise based clusterings contain significantly more and smaller clusters than the modularity based ones, being \emph{refinements} of the latter in the case of \texttt{karate} and \texttt{lesmis}.
Another striking observation is that the surprise based clusterings contain far more \emph{singletons}, i.e. clusters containing only one vertex with usually low degree; this can be explained by the fact that surprise does not take vertex degrees into account and hence, merging low degree vertices into larger clusters causes larger penalties.
It reconstructs the ground-truth clustering of the \texttt{football} graph quite well.
This confirms the observations of Aldecoa and Mar{\'\i}n based on heuristically found clusterings~\cite{am-dncss-11}; in fact, we can show that for \texttt{karate}, this clustering was already optimal.
\section{Conclusion}
\label{sec:conclusion}
We showed that the problem of finding a clustering of a graph that is optimal with respect to the measure surprise is $\mathcal{NP}$-hard.
The observation that surprise is pareto optimal with respect to (maximizing) the number of edges and (minimizing) the number of vertex pairs within clusters yields a (polynomial time) dynamic program on trees.
Furthermore, it helps to find exact solutions in small, general graphs via a sequence of ILP computations.
The latter can be used to gain insights into the behavior of surprise, independent of any artifacts stemming from a particular heuristic.
Moreover, optimal solutions are helpful to assess and validate the outcome of heuristics.

\bibliographystyle{abbrv}
\bibliography{myRefs}
  
 \clearpage 
\appendix

\section{Proof of Lemma~\ref{lem:os_few_intra_non}}
\label{app:asymptotics}
The proof of Lemma~\ref{lem:os_few_intra_non} is based on the following two observations on the asymptotic behavior of binomial coefficients.
 \begin{lemma}
  Let $f : \mathbb{N} \rightarrow \mathbb{N}$ and $g : \mathbb{N} \rightarrow \mathbb{N}$ be two functions such that $g(n) \in o\left(f(n)\right)$. Then,
  $$\binom{f(n)}{g(n)} \in \Omega\left( \frac{f(n)^{g(n)}}{g(n)^{g(n) + 1/2}}\right)$$
  \label{lem:asymptotics}
 \end{lemma}
\begin{proof}
 For $n > 0$, Stirling's formula yields
 $$  \sqrt{2 \pi} \cdot n^{n + 1/2} \cdot e^{-n} \leq n! \leq e \cdot \sqrt{2\pi} \cdot n^{n + 1/2} \cdot e^{-n}$$
 Hence, it is
 \begin{align*}
  \binom{f(n)}{g(n)}&= \frac{f(n)!}{g(n)! \cdot [f(n) - g(n)]!}\\
		    &\geq \frac{\sqrt{2 \pi} \cdot {f(n)}^{f(n)} \cdot \frac{\sqrt{f(n)}}{e^{f(n)}}}{(e \cdot \sqrt{2 \pi})^2 \cdot {g(n)}^{g(n)} \cdot \frac{\sqrt{g(n)}}{e^{g(n)}} \cdot  {[f(n) - g(n)]}^{f(n) - g(n)} \cdot \frac{\sqrt{f(n) - g(n)}}{e^{f(n) - g(n)}}}\\
		    &= \frac{1}{e^2 \cdot \sqrt{2\pi}} \cdot \frac{ {f(n)}^{f(n)} \cdot \sqrt{f(n)}}{ {g(n)}^{g(n)} \cdot \sqrt{g(n)} \cdot  {[f(n) - g(n)]}^{f(n) - g(n)} \cdot \sqrt{f(n) - g(n)}}\\
 \end{align*}
As $f(n)$ grows faster than $g(n)$, 
$$\sqrt{\frac{f(n)}{g(n) \cdot (f(n) - g(n))}} \in \Theta\left(\frac{1}{\sqrt{g(n)}}\right)$$
and thus,
 \begin{align*}
  \binom{f(n)}{g(n)}&\in \Theta \left(  \frac{1}{e^2 \cdot \sqrt{2\pi}} \cdot \frac{ {f(n)}^{f(n)} }{ {g(n)}^{g(n)} \cdot \sqrt{g(n)} \cdot  {[f(n) - g(n)]}^{f(n) - g(n)} }\right)
 \end{align*}
 It is $f(n)-g(n) \leq f(n)$ and hence,
  \begin{align*}
  \binom{f(n)}{g(n)}&\subseteq \Omega\left(\frac{{f(n)}^{g(n)}}{\sqrt{g(n)} \cdot {g(n)}^{g(n)}}\right) = \Omega\left( \frac{{f(n)}^{g(n)}}{{g(n)}^{g(n) + 1/2}}\right)
 \end{align*} \qed
\end{proof}
\begin{lemma}
 \label{lem:As_Pol}
 Let $u_1, u_2, k_1, k_2 \in \mathbb{N}$ with $u_1 > k_1$, $u_2 > k_2$ and $k_1 > k_2$.
 Furthermore, let $f_1 : \mathbb{N}\rightarrow\mathbb{N}$, $f_2 : \mathbb{N}\rightarrow\mathbb{N}$, $g_1 : \mathbb{N}\rightarrow\mathbb{N}$ and $g_2 : \mathbb{N}\rightarrow\mathbb{N}$ be functions with $f_1(n) \in \Theta(n^{u_1})$, $f_2(n) \in \Theta(n^{u_2})$, $g_1(n) \in \Theta(n^{k_1})$ and $g_2(n) \in \Theta(n^{k_2})$.
 Then,
 $$\binom{f_2(n)}{g_2(n)} \in o\left(\binom{f_1(n)}{g_1(n)}\right)$$
\end{lemma}
\begin{proof}
 From Lemma~\ref{lem:asymptotics}, it follows that
 $$\binom{f_1(n)}{g_1(n)} \in \Omega\left(\frac{f_1(n)^{g_1(n)}}{\sqrt{g_1(n)} \cdot {g_1(n)}^{g_1(n)}}\right)$$
 Furthermore, for large $n$ there exist constants $a_1, b_1, b_2 > 0$ such that
  \begin{itemize}
  \item $b_1 \cdot n^{k_1} \leq g_1(n) \leq b_2 \cdot n^{k_1} $
  \item $a_1 \cdot n^{u_1} \leq f_1(n)$
 \end{itemize}
 and hence, as $f_1(n)/g_1(n) \geq 1$ for large $n$,
 \begin{align*}
  \frac{f_1(n)^{g_1(n)}}{\sqrt{g_1(n)} \cdot {g_1(n)}^{g_1(n)}} &= \frac{1}{\sqrt{g_1(n)}} \cdot {\left(\frac{f_1(n)}{g_1(n)}\right)}^{g_1(n)} \geq \frac{1}{\sqrt{g_1(n)}} \cdot {\left(\frac{f_1(n)}{g_1(n)}\right)}^{b_1 \cdot n^{k_1}}\\
                                                                &\geq \frac{1}{\sqrt{b_2} \cdot n^{1/2 k_1}} \cdot \frac{{\left(a_1 \cdot n^{u_1}\right)}^{b_1 \cdot n^{k_1}}}{{\left(b_2 \cdot n^{k_1}\right)}^{b_1 \cdot n^{k_1}}}
 \end{align*}
From this, it follows that
$$\binom{f_1(n)}{g_1(n)} \in \Omega\left(\frac{{a_1}^{b_1 \cdot n^{k_1}} \cdot n^{b_1 \cdot u_1 \cdot n^{k_1}}}{n^{1/2 k_1} \cdot {b_2}^{b_1 \cdot n^{k_1}} \cdot n^{b_1 \cdot k_1 \cdot n^{k_1}}} =: l_1(n)\right)$$
On the other hand there exist constants $a_2, b_3 > 0$ such that for large $r$
\begin{itemize}
 \item $f_2(n) \leq a_2 \cdot n^{u_2}$
 \item $g_2(n) \leq b_3 \cdot n^{k_2}$
\end{itemize}
and hence,
$$\binom{f_2(n)}{g_2(n)} \leq {f_2(n)}^{g_2(n)} \leq (a_2 \cdot n^{u_2})^{b_3 \cdot n^{k_2}} = {a_2}^{b_3 \cdot n^{k_2}} \cdot n^{b_3 \cdot u_2 \cdot n^{k_2}} =: l_2(n)$$
It remains to show that $l_2(n) \in o(l_1(n))$. 
To see that this is the case, we look at the logarithm:
\begin{align*}
  \log(l_1(n)) &= b_1 \cdot n^{k_1} \cdot \log(a_1) + b_1 \cdot u_1 \cdot n^{k_1} \cdot \log(n) - \frac{1}{2} \cdot k_1 \cdot \log(n) \\
  &~~~~~ - b_1 \cdot n^{k_1} \cdot \log(b_2) - b_1 \cdot k_1 \cdot n^{k_1} \cdot \log(n)\\
              &= b_1\cdot\underbrace{(u_1 - k_1)}_{>0} \cdot n^{k_1} \cdot \log(n) + b_1 \cdot(\log(a_1) - \log(b_2)) \cdot n^{k_1}\\
              &~~~~~ - \frac{1}{2} \cdot k_1 \cdot \log(n) 
\end{align*}
Hence, $\log(l_1(n)) \in \Theta(n^{k_1} \cdot \log(n))$.
On the other hand,
$$\log(l_2(n)) = b_3 \cdot n^{k_2} \cdot \log(a_2) + b_3 \cdot u_2 \cdot n^{k_2} \cdot \log(n) \in \Theta(n^{k_2} \cdot \log(n))$$
Thus, $l_2(n) \in o(l_1(n))$. \qed

\end{proof}
We are now ready to proof Lemma~\ref{lem:os_few_intra_non}.
\begin{proof}[of Lemma~\ref{lem:os_few_intra_non}]
 Assume that $\C$ is an optimal clustering with respect to surprise and $i_p(\C) - i_e(\C) >  \frac{r^4}{2}$.
 We will compare $S(\C)$ to the value of the clustering $\C'$ used in the proof of Lemma~\ref{lem:os_many_ie}.
 $\C'$ is a clustering into cliques with $|\Set| \cdot \binom{r^2}{2}$ intracluster edges.
 Hence, $i_p(\C')=i_e(\C')$ and thus
 \begin{align*}
  \binom{p}{m} \cdot S(\C')&=\binom{i_p(\C')}{i_e(\C')} \cdot \binom{p-i_e(\C')}{m-i_e(\C')} = \binom{\binom{|\Set| \cdot r^2 + |\X|}{2} - |\Set| \cdot \binom{r^2}{2} := f_2(r)}{|\X| \cdot r^2 + \binom{|\X|}{2} := g_1(r)}
 \end{align*}
 with $f_2 \in \Theta(r^6)$ and $g_2 \in \Theta(r^3)$.
 
  \emph{Case 1: $i_e(\C) \leq i_p(\C)/2$}: As $i_e(\C) \leq i_p(\C)/2$, substituting a lower bound for $i_e(\C)$ decreases $\binom{i_p(\C)}{i_e(\C}$.
  From Lemma~\ref{lem:os_many_ie}, we know that $i_e(\C) \geq |\Set| \cdot \binom{r^2}{2}$, which can be estimated from below by $r^4$ for large $r$.
  Altogether, we get that 
    $$\binom{p}{m} \cdot S(\C) =\!\! \sum_{i=i_e(\C)}^{m} \binom{i_p(\C)}{i} \cdot \binom{p-i_p(\C)}{m-i} \!\geq \!\binom{i_p(\C)}{i_e(\C)} \geq \binom{2 \cdot |\Set| \cdot \binom{r^2}{2} := f_1(r)}{r^4 := g_1(r)}$$
  with $f_1 \in \Theta(r^5)$ and $g_1 \in \Theta(r^4)$.
  Now we can use Lemma~\ref{lem:As_Pol} to see that, for large $r$, $S(\C)$ is larger than $S(\C')$, which contradicts the optimality of $\C$.
  
  \emph{Case 2: $i_e(\C) > i_p(\C)/2$}: We have $\binom{i_p(\C)}{i_e(\C)} = \binom{i_p(\C)}{i_p(\C) - i_e(\C)}$.
  As $i_e(\C) > i_p(\C)/2$, $i_p(\C) - i_e(\C) < i_p(\C)/2$ and substituting a lower bound for $i_p(\C)-i_e(\C)$ decreases $\binom{i_p(\C)}{i_p(\C) -i_e(\C)}$.
  Hence, by Lemma~\ref{lem:os_few_intra_non},
  $$\binom{p}{m} \cdot S(\C) \geq \binom{i_p(\C)}{i_e(\C)} = \binom{i_p(\C)}{i_p(\C) - i_e(\C)} \geq \binom{i_p(\C)}{r^4/2} \geq \binom{|\Set| \cdot \binom{r^2}{2} := f_1(r)}{r^4/2 := g_1(r)}$$
  with $f_1 \in \Theta(r^5)$ and $g_1 \in \Theta(r^4)$.
  Analogously to Case 1, it follows that $\C$ was not optimal. \qed
 \end{proof}

\section{Proof of Theorem 14}
\label{app:trees}
In this section, we describe a straigtforward generalization of the dynamic program of Dinh and Thai~\cite{dt-tocdf-} to arbitrary vertex weights and explain in detail, how this can be used to solve surprise minimization in trees in polynomial time.

Let $T = (V,E)$ be a (rooted) tree with root $r$ together with a function $w:V\rightarrow \mathbb{Q}_{\geq 0}$ which assigns a weight to each vertex of the tree. 
In a natural way, $w(C) = \sum_{v \in C}{w(v)}$ denotes the weight of a subset $C \subseteq V$ of vertices.
$n:= \left|V\right|$ denotes the number of vertices in $T$ and  $m:= \left|E\right| = n-1$ the number of edges.
$T^u = (V^u, E^u)$ is the subtree of $T$ which is rooted at node $u$ and $u_1, \ldots, u_{t(u)}$ are the children of node $u$.
For $i=0,\ldots,t(u)$, let $T^u_{i}$ be the partial subtree of $T$ rooted at node $u$ and consisting of $u, T^{u_1}, \ldots, T^{u_i}$.
Fig.~\ref{partialTree} illustrates the concept of partial subtrees.\\
\begin{figure}[tbp]
\centering
\includegraphics[width=0.5\textwidth]{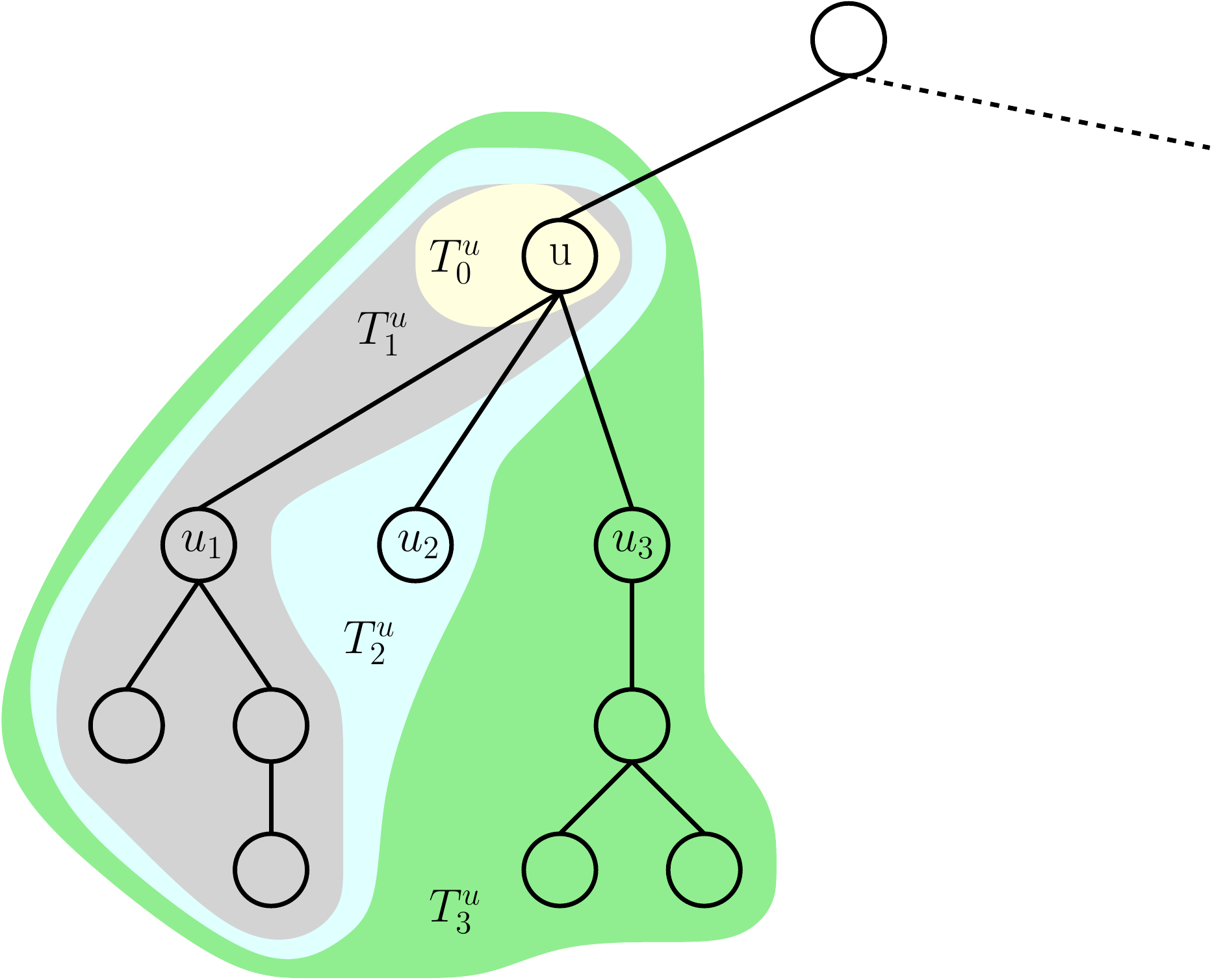}
\caption{partial subtrees of a node u.}
\label{partialTree}
\end{figure}

An important observation is that in an optimal solution of the MACP problem, all clusters are connected. Hence, in a clustering $\mathcal{C}$ of a tree $T$ with $c$ clusters there are exactly $c-1$ intercluster and $m-c+1$ intracluster edges. According to this observation, the following functions are the core of the dynamic program.
\begin{itemize}
	\item $F^{u}(k)$: The minimum of the sum-of-squares of component-weights in the subtree $T^u$ when $k$ edges are removed from $T^u$.
	\item $F^{u}(k, \nu)$: The minimum of the sum-of-squares of component-weights in the subtree $T^u$ when $k$ edges are removed from $T^u$ while the component which contains $u$ has weight $\nu$.
	\item $F^{u}_i(k, \nu)$: The minimum of the sum-of-squares of component-weights in the partial subtree $T^u_{i}$ when $k$ edges are removed in $T^u_i$ and the component that contains $u$ has weight $\nu$.
\end{itemize}

It is easy to see that these functions are related as follows:
\begin{align*}
	F^u(k, \nu) &= F^u_{t(u)} (k, \nu) \\
	F^u(k) &= \min_{w(u) \leq \nu \leq w(T^u)} F^u(k, \nu)
\end{align*}
Therefore it is only required to have a look at the calculation of $F^u_i$. The basic cases are the following:

\begin{align}
\label{k0}
F^{u}_{i}(0,\nu) &= \Bigg\{\begin{array}{ll}
						w(T^{u}_i)^2, & \text{for } \nu = w(T^{u}_i) \\ 					\infty, & \text{otherwise }
					\end{array}\\
\label{i0}
F^{u}_{0}(k,\nu) &= \Bigg\{\begin{array}{ll}
						w(u)^2, & \text{for } \nu = w(u)\\
						\infty, & \text{otherwise } 
					\end{array}
\end{align}

Starting with the leaves of the tree $T$ the function $F^u_i$ is computed in the following recursive way.

\begin{equation}
F^{u}_{i}(k,\nu) = \min \left\{\begin{array}{ll}
						\displaystyle \min_{0 \leq l \leq k-1} \{ F^{u}_{i-1}(l, \nu) + F^{u_i}(k-l-1) \}, \\
						\displaystyle \min_{0 \leq l \leq k, 0 \leq \mu \leq \nu} \{ F^{u}_{i-1}(l, \mu) + F^{u_i}(k-l, \nu - \mu) + 2 \mu (\nu - \mu) \} 
					\end{array}\right\} \\
					\label{eq:fi}
\end{equation}

Equation (\ref{eq:fi}) contains two cases that occur when $k$ edges are removed from the partial subtree $T^u_i$:
\begin{itemize}
	\item When the edge $\{u, u_i\}$ is removed and there are $l$ edges removed in $T^u_{i-1}$, then only $k-l-1$ edges can be removed in the subtree $T^u_i$. Also the component which contains node $u$ has the same weight as the component of $T^u_{i-1}$ that contains $u$.
	\item When the edge $\{u, u_i\}$ is not removed the weight of the component containing $u$ differs from the weight in $T^u_{i-1}$. When this component in $T^u_{i-1}$ has weight $\mu$ then the component containing $u$ in $T^{u_i}$ has weight $\nu - \mu$. Thus the factor $2\mu(\nu - \mu) = \nu^2 - \mu ^2 - (\nu - \mu)^2$ provides the correct weight for the new component that contains $u$.
\end{itemize}

Algorithm \ref{treeAlgorithm} shows how these equations can be used to calculate the surprise value of an optimal clustering with respect to surprise with the help of a dynamic program. 

\begin{algorithm}[tbp!]
\caption{Surprise Minimization on Trees}
\label{treeAlgorithm}
\DontPrintSemicolon
\ForEach{$u \in V$}{
	$w(u) = 1$\;
}
\ForEach{$u \in V$ in topological order}{
	\For{$i = 0$ to $t(u)$}{
		\For{$k = 0$ to $\left| E^u \right|$}{
			\For{$\nu = 0$ to $w(T^u)$}{
				Compute $F^{u}_{i}(k,\nu)$, $F^{u}(k,\nu)$ and $F^{u}(k)$\;
			}
		}
	}
}

\Return $\min_{0 \leq k \leq m}{S(m-k, \frac{1}{2} F^r(k) - \frac{1}{2} n)}$\;
\end{algorithm}

\begin{theorem}
 Algorithm \ref{treeAlgorithm} computes the surprise value of an optimal clustering with respect to surprise in $O(n^5)$ time.
\end{theorem}

\begin{proof}
Consider a solution of MACP on a tree deleting $k$ edges with value $F^r(k)$. 
The induced clustering has $i_e = m-k$ edges inside clusters and the number of intracluster pairs is $i_p = \frac{1}{2} F^r(k) - \frac{1}{2} n$ (see Eq. \eqref{minIP_MACP} in Sect.~\ref{sec:complexity}).
Hence, an optimal clustering of \mipShort~with parameter $i_e = m-k$ corresponds to an optimal clustering of MACP with parameter $k$ and vice versa. Thus, the optimal surprise value of a clustering for fixed $k$ is $S(m-k, \frac{1}{2} F^r(k) - \frac{1}{2} n)$.
As explained in Sect.~\ref{sec:connection}, the global optimum can then be found by checking all possibilities for $0 \leq k \leq m$.

A lookup on every node in $T$ and its children can be done in $O(n)$. There are $O(w(T^u)^2)$ possible combination of the variables $k$ and $\nu$ for each node $u$ and its subtree $T^u$. Further, $w(T^u)$ is trivially bounded by $n$. The computation of one $F^u_i(k,\nu)$ can be done in $O(n^2)$ and thus, the overall complexity of Algorithm \ref{treeAlgorithm} is in $O(n^5)$. \qed
\end{proof}

Algorithm \ref{treeAlgorithm} only returns the value of an optimal clustering, not the clustering itself. Nevertheless, it can be modified in a straightforward way to return an optimal clustering instead of its surprise value without a loss in running time, which yields Theorem \ref{theoremTrees}.

\section{Heuristics for Linear Programs}
  \label{app:heuristics}
  We tried the following modifications to further decrease the running time to compute exact solutions:
  \begin{itemize}
  \item \emph{Prune small $k$ (PSK)}: We first determine the clustering into cliques that maximizes the number $k_\mathrm{start}$ of intracluster edges.
    This can be done by dropping~(\ref{con:exactly_k}), substituting~(\ref{eq:basic_objective}) by
  \begin{equation}
    \text{maximize }\sum_{\{u,v\} \in E} \Xer{u}{v}
  \end{equation}
  and setting $\Xer{u}{v}=0$ for all vertex pairs $\{u,v\}$ not connected by an edge.
  Proposition~\ref{lem:props_opt_sol}(iii) then yields that we do not have to consider clusterings with less than $k_\mathrm{start}$ intracluster edges.
       This is in fact a special case of the gap variant, but as solving the modified ILP is usually very fast, its usage potentially decreases the overall running time for all variants. 
  \item \emph{Testing for Feasibility (TF)}:
  From the value $S$ of the best current solution, we can compute for each $k$ the largest $i_p$ such that $S(i_p,k) < S$.
  This can be modeled as an additional constraint; if this makes the model infeasible, we can safely proceed to the next $k$.
  The downside of this approach is that the lower bounds for the gap and relaxed variant are updated less often.
  However, it potentially decreases the time to solve individual ILPs in case the model is not feasible.
  
  \item \emph{Enforce many intraedges (EMI)}:
    To enforce that the clustering we obtain by the linear program for the relaxed variant has the most intracluster edges among all valid clusterings that minimize the number of intracluster pairs, and therefore yields the best upper bound, we replace (\ref{eq:basic_objective}) by 
  \begin{equation}
    \text{minimize } \ m \cdot \left(\sum_{\{u,v\} \in \binom{V}{2}} \Xer{u}{v}\right) - \sum_{\{u,v\} \in E} \Xer{u}{v}
  \end{equation}
    Similarly, for the gap variant, we replace~(\ref{eq:gap_objective}) by
       \begin{equation}
    \text{minimize } \ m \cdot \left(\sum_{\{u,v\} \in \binom{V}{2}} \Xer{u}{v} - \sum_{\{u,v\} \in E} \Xer{u}{v}\right) - \sum_{\{u,v\} \in E} \Xer{u}{v} 
  \end{equation}
    Obviously, this does not make sense for the exact variant.
  \end{itemize}
  
  Table~\ref{tab:running_times_modifications} shows an overview of running times and the number of solved ILPs of the different strategies on the test instances from Sect.~\ref{sec:algorithms}.
  \begin{table}[tbp]
   \caption{Running times in seconds of successive ILP approach, different strategies.}
   \label{tab:running_times_modifications}
   \begin{center}
   \begin{footnotesize}
\begin{tabular}{|cccc|cc|cc|cc|cc|} 
\hline \multicolumn{4}{|c|}{} & \multicolumn{2}{c|}{lesmis} & \multicolumn{2}{c|}{karate} & \multicolumn{2}{c|}{grid6} & \multicolumn{2}{c|}{dolphins} \\ 
var & TF & PSK & EMI & ILP & t(s) & ILP & t(s) & ILP & t(s) & ILP & t(s) \\ 
\hline e & n &n & - & 255 & 1194 & 79 & 51 & 61 & 470 & 160 & 497 \\ 
e & n &y & - & 119 & 1149 & 54 & 50 & 43 & 470 & 104 & 489 \\ 
e & y &n & - & 255 & 315 & 79 & 15 & 61 & 689 & 160 & 164 \\ 
e & y &y & - & 119 & 272 & 54 & 14 & 43 & 684 & 104 & 152 \\ 
r & n &n & n & 176 & 283 & 49 & 21 & 42 & 449 & 107 & 163 \\ 
r & n &n & y & 176 & 373 & 49 & 31 & 42 & 2091 & 107 & 383 \\ 
r & n &y & n & 41 & 254 & 25 & 20 & 25 & 448 & 52 & 158 \\ 
r & n &y & y & 41 & 353 & 25 & 30 & 25 & 2091 & 52 & 378 \\ 
r & y &n & n & 254 & 302 & 78 & 15 & 60 & 1154 & 159 & 218 \\ 
r & y &n & y & 254 & 354 & 78 & 15 & 60 & 1698 & 159 & 537 \\ 
r & y &y & n & 119 & 264 & 54 & 15 & 43 & 1129 & 104 & 212 \\ 
r & y &y & y & 119 & 323 & 54 & 14 & 43 & 1700 & 104 & 539 \\ 
g & n &n & n & 112 & 206 & 39 & 15 & 37 & 402 & 91 & 147 \\ 
g & n &n & y & 23 & 165 & 18 & 18 & 15 & 1904 & 45 & 131 \\ 
g & n &y & n & 29 & 186 & 20 & 15 & 20 & \textbf{398} & 50 & 143 \\ 
g & n &y & y & 23 & \textbf{162} & 18 & 18 & 15 & 1896 & 45 & 131 \\ 
g & y &n & n & 195 & 259 & 73 & 15 & 56 & 1774 & 144 & 168 \\ 
g & y &n & y & 107 & 221 & 52 & \textbf{12} & 38 & 2224 & 100 & 110 \\ 
g & y &y & n & 112 & 245 & 54 & 14 & 39 & 1724 & 103 & 160 \\ 
g & y &y & y & 107 & 222 & 52 & \textbf{12} & 38 & 2214 & 100 & \textbf{108} \\ 
\hline \end{tabular} 
  \end{footnotesize}
  \end{center}
  \end{table}
 
In almost all cases, the PSK heuristic is able to decrease the running time slightly.

Enforcing many intracluster edges always increased the running time of the relaxed variant; the average running time for each linear program increases and in none of our examples it helped to decrease their number.
For the gap variant, this modification was beneficial in most cases.
However, for \texttt{grid6}, the running time increased by almost a factor of five compared to the gap variant without modifications.

Similarly, testing for feasibility is beneficial in combination with the gap and relaxed variant in about half of the cases, but on \texttt{grid6}, it increases their running time significantly.

Overall, the unmodified version of the gap variant was always faster than any version of the relaxed or exact one.
Among the versions of the gap variant, the one that uses only PSK and the one with all modifications exhibit good overall behavior, while the former seems to be more robust.

\end{document}